\documentclass[11pt,letterpaper]{article}

\usepackage{nicefrac}
\usepackage{algorithm2e}

\usepackage[margin=1in]{geometry}
\usepackage{amsmath}
\usepackage{amsthm}
\usepackage{amssymb}
\usepackage{color}
\usepackage{hyperref}
\usepackage{fancybox}
\usepackage{bm}
\usepackage{bbm}

\newtheorem{lemma}{Lemma}

\newenvironment{fminipage}
  {\begin{Sbox}\begin{minipage}}
  {\end{minipage}\end{Sbox}\fbox{\TheSbox}}

\newenvironment{algbox}[0]{\vskip 0.05in
\noindent 
\begin{fminipage}{5.4in}
}{
\end{fminipage}
\vskip 0.05in
}

\newcommand{\defeq}{:=}
\newcommand{\one}{\mathbbm{1}}
\newcommand{\logs}{\widehat{\log^{*}}}
\newcommand{\logss}{\widehat{\log^{**}}}
\newcommand{\sizemax}{\textsc{Size}^{\max}}
\newcommand{\size}{\textsc{Size}}
\newcommand{\parent}{\textsc{Parent}}

\begin{document}

\title{Simpler Analyses of Union-Find}

\author{
Zhiyi Huang\\
The University of Hong Kong\\
\texttt{zhiyi@cs.hku.hk}
\and 
Chris Lambert\\
Carnegie Mellon Univeristy\\
\texttt{cslamber@andrew.cmu.edu}
\and
Zipei Nie\\
LMCRC, Huawei
\footnote{Part of this work was done when the author visited Institut des Hautes Études Scientifiques.}\\
\texttt{niezipei@huawei.com}
\and
Richard Peng\\
Carnegie Mellon University
\footnote{Part of this work was done when the author was at the University of Waterloo.}\\
\texttt{yangp@cs.cmu.edu}
}

\begin{titlepage}
\thispagestyle{empty}

\maketitle

\begin{abstract}
\thispagestyle{empty}
We analyze union-find using potential functions motivated by continuous algorithms, and give alternate
proofs of the $O(\log\log{n})$, $O(\log^{*}n)$,
$O(\log^{**}n)$, and $O(\alpha(n))$ amortized cost upper bounds. The proof of the $O(\log\log{n})$ amortized bound goes as follows.
Let each node's potential be the square root of its size, i.e., the size of the subtree rooted from it.
The overall potential increase is $O(n)$ because the node sizes increase geometrically along any tree path.
When compressing a path, each node on the path satisfies that either its potential decreases by $\Omega(1)$, or its child's size along the path is less than the square root of its size:
this can happen at most $O(\log\log{n})$ times along any tree path.
\end{abstract}
\end{titlepage}

\section{Introduction}

The union-find, or disjoint set, data structure maintains disjoint sets of elements under modifications that \emph{union} two of the sets, and answers queries for \emph{finding} the set containing a given element.
It is a key primitive in many efficient algorithms~\cite{H85,FGMT13,LS13,DBS18},
and widely taught in undergraduate courses on algorithms.
Furthermore, analyses of union-find have been instrumental to developing amortized analysis and data structures with low amortized costs~\cite{GF64,HU73,Tarjan75,T79,FS89,GI91}.

Despite union-find being taught in almost every undergraduate algorithms curriculum,
proving an $O(\log^{*}n)$ upper bound
for it in class is often considered ambitious.
Doing so often entails defining bucketing schemes for the elements~\cite{G19}.
The tight $O(\alpha(n))$ bound,
where $\alpha$ is the inverse Ackermann function,
is often discussed at a high level in theory-focused courses.

We reexamine the analyses of union-find from two starting points.
First, the $\log(\size(p))$ potential function
(where $\size(p)$ is the size of the subtree rooted at $p$)
that's present in analyses of many tree-based data structures~\cite{ST85,FSST86,ST86},
including path-compression without union-by-size.
Second, recent developments in continuous algorithms and optimization often revolve around the design of more sophisticated potential functions.
We combine these to give an amortized analysis of union-find via potential functions that are functions of subtree sizes.
Compared to other analyses of union-find, these potential functions
differ in that they are naturally continuous, and do not involve
discretization/bucketing.

We work with the union-by-size version due to its closer
connection with the $\log(\size(p))$ potential function.
Section~\ref{sec:prelim} formalizes the data structure
and the framework of potential function analysis for union-find.
Sections~\ref{sec:loglog} and \ref{sec:logstar} then show that analyzing union-find using the potential
functions $\sqrt{\size(p)}$ and $\frac{\size(p)}{\log^2 \size(p)}$
give amortized costs of $O(\log\log{n})$ and $\log^{*}n$ respectively.
Section~\ref{sec:extensions} builds on these analyses to give a tighter bound of
$O(\log^{**}n)$ and sketches how to go beyond.
Such analysis requires a slightly lesser-known characterization
of the inverse Ackermann function~\cite{S06}:
we include its proof in Appendix~\ref{sec:ackermann} for completeness.
With these definitions, we also provide amortized analyses
of the $O(\alpha(n))$ amortized cost bound using the more classical recursive definitions in Section~\ref{sec:direct}.
Section~\ref{sec:conclusion} discusses our results and avenues for further simplifications.

\section{Preliminaries}
\label{sec:prelim}

Logarithms are base $2$ unless stated otherwise;
we use $\ln(x)$ for the natural logarithm.
Numerical superscripts in parentheses mean repeated applications of functions, i.e.
\[
    f^{(k)}(x) = \underbrace{f( f( \dots f}_{\mbox{\footnotesize $k$ times}}(x) \dots ) )
    ~.
\]

\subsection{Union-Find}
The union-find data structure represents disjoint sets of elements using a rooted forest.
The nodes are the elements.
Each rooted tree in the forest corresponds to a set, with the root being the representative element.
Finding the representative of a node $p$, which we will refer to as $p$'s root,
consists of traversing the path from $p$ to the root, which we will refer to as a find-path.
Studies of such structures and their heuristics,
specifically union-by-size,
started as early as the 1960s~\cite{GF64}.

The most efficient version of union-find,
first studied by Tarjan~\cite{Tarjan75},
rely on two manipulations of parent pointers:
path compression, which points everything on a find-path to the root,
and union-by-size, which points from the root of the smaller tree to the root of the larger one when we union two sets.
Figure~\ref{fig:union-by-size} presents the pseudocode of union-by-size.

\begin{figure}
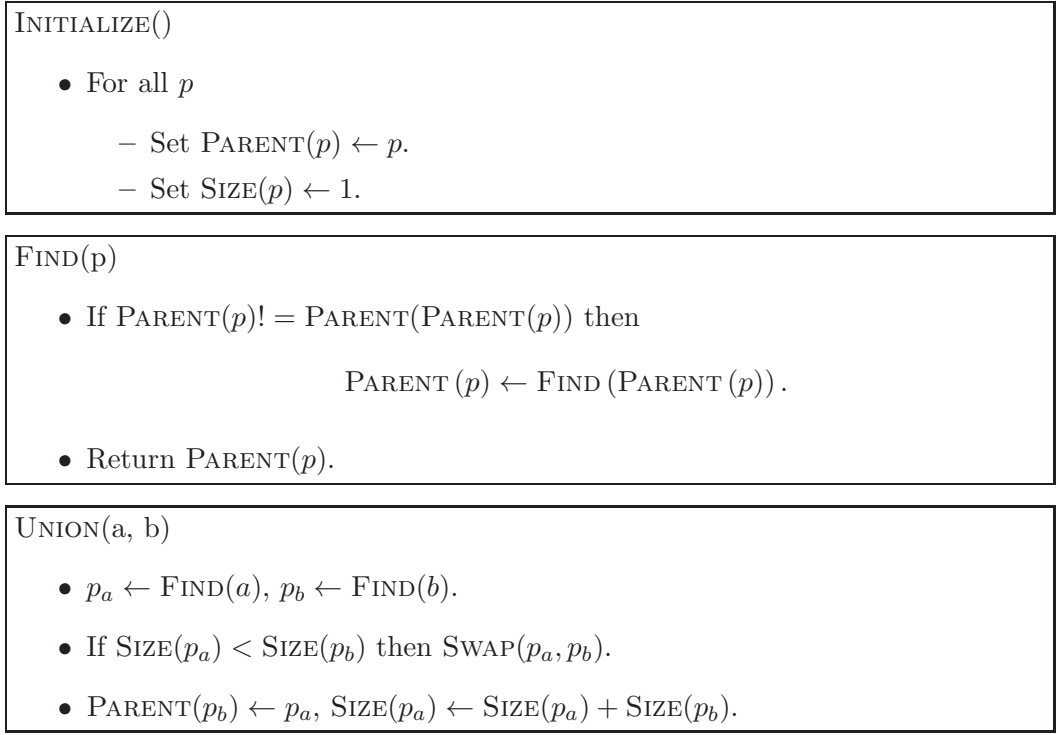

\centering
\begin{algbox}
\textsc{Initialize}()
\begin{itemize}
\item For all $p$
\begin{itemize}
    \item Set $\parent(p) \leftarrow p$.
    \item Set $\size(p) \leftarrow 1$.
\end{itemize}
\end{itemize}
\end{algbox}

\begin{algbox}
\textsc{Find}(p)
\begin{itemize}
\item If $\parent(p) != \parent(\parent(p))$ then
\[
\parent\left(p\right)
\leftarrow
\textsc{Find}\left(\parent\left(p\right)\right).
\]
\item Return $\parent(p)$.
\end{itemize}
\end{algbox}

\begin{algbox}
\textsc{Union}(a, b)
\begin{itemize}
\item $p_a \leftarrow \textsc{Find}(a)$,
 $p_b \leftarrow \textsc{Find}(b)$.
 \item If $\size(p_a) < \size(p_b)$ then $\textsc{Swap}(p_a, p_b)$.
 \item $\parent(p_b) \leftarrow p_a$,
 $\size(p_a) \leftarrow \size(p_a) + \size(p_b)$.
\end{itemize}
\end{algbox}
\caption{Pseudocode for Union-by-Size}
\label{fig:union-by-size}
\end{figure}

The nodes' sizes evolve over time:
root nodes may acquire new descendants in union steps, while non-root nodes may lose descendants in find steps because of path compressions.
We consider two types of size functions that may differ for non-root nodes. 
At any time, the \emph{current size} of a node $p$, denoted as $\size(p)$, is its current number of descendants.
The \emph{max size} of $p$, denoted as $\sizemax(p)$, is the number of nodes that ever become a descendant of $p$ at some point of time. $\size(p)$ evolves over time and is maintained by the algorithm, while $\sizemax(p)$ stays invariant and is only used in the analysis. The current sizes will be used in Sections \ref{sec:loglog} and \ref{sec:logstar}, while the max sizes will be used in Sections \ref{sec:extensions} and \ref{sec:direct}.

We next derive some properties of these size functions. First, we study the monotonicity of $\size(p)$.

\begin{lemma}\label{lem:monotonicity}
$\size(p)$ is nondecreasing in steps that keep $p$ as a root, and is nonincreasing in the other steps.
\end{lemma}

\begin{proof}
    In a find step, a node $p$'s set of descendants changes (in fact, becomes smaller) if and only if it is an interior node of the fine-path.
    In a union step, a node $p$'s set of descendants changes (in fact, becomes larger) if and only if it is the new root.
    Hence, for any step that keeps $p$ as a root, $p$ cannot be an interior node of a find-path, and thus its set of descendants either stays the same or becomes larger. 
    In any other step, $p$ cannot be the new root of a union step, and thus its set of descendants either stays the same or becomes smaller.
\end{proof}

The trees are reconfigured via changes to $\parent(p)$.
We next prove the monotonicity of $\sizemax(\parent(p))$ under these parent changes,
and that the trees are balanced with respect to these max sizes.

\begin{lemma}\label{lem:parent-size-monotone}
    For any non-root node $p$,
    the max-size of its parent,
    $\sizemax(\parent(p))$ is nondecreasing over time.
\end{lemma}

\begin{lemma}\label{lem:geometric-decrease}
    At any time, for any non-root node $p$ and its parent at the time $\parent(p)$, we have
    \[
        \sizemax(\parent( p) ) \ge 2 \cdot \sizemax(p)
        ~.
    \]
\end{lemma}
\begin{proof}[Proof of Lemmas~\ref{lem:parent-size-monotone} and \ref{lem:geometric-decrease}]
    We prove the lemmas by induction on time.
    Initially, the lemmas hold vacuously because all nodes are roots.
    
    In a union step, a node $p$ becomes the child of the new root $r = \parent(p)$. 
    We have
    \[
        \size(\parent(p)) \ge 2 \cdot \size(p)
    \]
    after the step by definition of union-by-size.
    Further, at this point, we have $\size(p) = \sizemax(p)$ because $p$ can no longer acquire new descendants as a non-root in the future.
    Further, we have $\sizemax(\parent(p)) \ge \size(\parent(p))$ by the definitions of current and max sizes.
    Hence, the lemma continues to hold.

    In a find step, path compression involving
    some non-root nodes $p$ may cause $p$'s parent
    to become the root node, $r$. 
    By the induction hypothesis that the lemma holds before the step, when $r$ was an ancestor of $p$, we get that $\sizemax(r) \ge 2 \cdot \sizemax(p)$.
    Hence, the lemma continues to hold.
\end{proof}

We remark that the balanced property of max-sizes
(Lemma~\ref{lem:geometric-decrease})
may not hold for current sizes.
This is because path compressions can
decrease some nodes' current sizes.
Consider for example a complete binary tree of height $3$. A path compression on the leftmost path leaves the root's original left child with only one grandchild, i.e., its original right child. That grandchild has size $3$, but its parent has size $4$.

Finally, as a corollary of max sizes' geometric decrease along tree paths (Lemma~\ref{lem:geometric-decrease}), we show that for slightly sublinear functions, the sum of function values for the max sizes is at most linear in the number of nodes. 

\begin{lemma}\label{lem:sizemax-upper-bound}
    We have      
    \[
        \sum_p \frac{\sizemax(p)}{(1+\log\sizemax(p))^2}=O(n)
    \] where $p$ ranges over all $n$ nodes.
\end{lemma}
\begin{proof}

    Distribute each node $p$'s contribution among all nodes that were once its descendants. That is, for each $q$ that was once $p$'s descendent, we charge
    \[
        \frac{1}{(1+\log\sizemax(p))^2}
    \]
    to $q$.
    We consider the total charged to some node $q$.

    Let $p_0 = q, p_1, \dots, p_k$ be the set of nodes that were once $q$'s ancestors, sorted by the time $t_i$ when $p_i$ is an ancestor of $q$.%
    \footnote{Since the set of ancestors of $q$ grows by at most one element in each step, the numbers $t_i$ ($1\le i\le k$) are distinct.}
    Because $p_{i-1}$ is an ancestor of $q$ at time $t_{i-1}$, the nodes $p_{i-1}$ and $q$ have the same root at time $t_i-1$. The next step must be a union step, and $p_{i-1}$ turns into a proper descendant of $p_i$ at time $t_i$. By Lemma~\ref{lem:geometric-decrease}, we have \[\sizemax(p_i)\ge 2\cdot\sizemax(p_{i-1})\] at time $t_i$. 
    Therefore, the total charged to $q$ is at most     
    \[
        \sum_{i = 1}^\infty \frac{1}{i^2} = O(1).
    \]
    Summing over all $n$ nodes then gives $O(n)$.
\end{proof}

\subsection{Potential Function Analysis}

We will design a potential function $\Phi(p) \ge 0$ for each node $p$, whose value depends on either the current size $\size(p)$ (Sections~\ref{sec:loglog} and~\ref{sec:logstar}) or the max sizes $\sizemax(p)$ and $\sizemax(\parent(p))$ (Sections~\ref{sec:extensions} and~\ref{sec:direct}).
We will then consider an overall potential as follows
\[
    \sum_{\mbox{\footnotesize non-root node $p$}} \Phi(p)
    ~.
\]

\begin{lemma}\label{lem:amortized-cost}
Let $h$ be a parameter (that can depend on $n$,
and $\Phi(p)$ be a potential function defined on the nodes
satisfies the following three properties:
    \begin{enumerate}
        \item (Monotonicity) $\Phi(p)$ is nonincreasing after $p$ becomes a non-root node;
        \item (Boundedness) At any time, for any node $p$, we have \[\frac{\Phi(p)}{h}=O\left(\frac{\sizemax(p)}{(1+\log\sizemax(p))^2}\right);\]
        \item (Amortized Path Length) In each find step, the find-path's length is at most $O(h)$ plus a constant times the decrease of the overall potential.
    \end{enumerate}
    Then, the total cost of $n$ union steps and $m$ find steps is at most $O \big( (m + n) \cdot h \big)$.
\end{lemma}

\begin{proof}
    The total cost of $n$ union steps is $O(n)$.
    It remains the analyze the total cost of $m$ find steps.

    By the monotonicity property, the overall potential can only increase due to a union step, after which we have a new non-root node $p$ that contributes to the summation.
    Further, by the boundedness property, the increase due to each node $p$ is at most:
    \[
        h \cdot O\left(\frac{\sizemax(p)}{(1+\log\sizemax(p))^2}\right)
    \]
    Hence, by Lemma~\ref{lem:sizemax-upper-bound}, the total increase of the overall potential in $n$ union steps is at most $O ( n h )$.

    Finally, by the amortized path length property, the total cost of $m$ find steps is upper bounded by $O(m h)$ plus a constant times the total decrease of the overall potential in $m$ find steps.
    Since the potential is initially $0$ and always non-negative, the latter part is at most the total increase of the overall potential in $n$ union steps, i.e., at most $O ( n h )$.

\end{proof}

We follow the convention of not explicitly tracking time in our notations.
However, because almost all of our analyses deal with how $\Phi(\cdot)$ changes over time,
it's worth remarking that $\parent(p)$, $\size(p)$, and $\Phi(p)$
are quantities that change over time,
while $\sizemax(p)$ is static over time.

\section{$O(\log\log{n})$ Analysis}
\label{sec:loglog}

Consider a potential function
\[
    \Phi \left( p \right)
    \defeq  
    \sqrt{\size\left(p\right)}
    ~.
\]

Since $\Phi(p)$ is an increasing function of $\size(p)$, the first property of potential function analysis holds by Lemma~\ref{lem:monotonicity}. The second property also holds because the exponent is less than $1$. 

We next verify the third property.
Consider a find step along $p_1 \to p_2 \to \cdots \to p_m \to r$, where $r$ is a root.
Suppose that the current sizes of $p_1, p_2, \dots, p_m$ are $s_1, s_2, \dots, s_m$. After the path compression, the size of $p_1$ stays the same, while the new sizes of $p_2, \dots, p_m$ are
\[
    s_{2} - s_1, s_{3} - s_2, \ldots, s_m - s_{m-1}
    ~.
\]

For every $2 \le i \le m$, the potential of $p_i$ decreases by $\sqrt{s_i} - \sqrt{s_i - s_{i-1}}$.
Either this is at least $\frac{1}{2}$, covering the cost of edge $p_{i-1} \to p_i$, or we have
\[
\sqrt{s_i - s_{i - 1}}
\geq
\sqrt{s_i} - \frac{1}{2}.
\]
Squaring both sides gives
\[
s_i - s_{i - 1}
\geq
s_i - \sqrt{s_i} + \frac{1}{4} > s_i - \sqrt{s_i}.
\]
Canceling $s_i$ on both sides and talking  logarithm gives
\[
    \log s_i \ge 2 \log s_{i-1}
    ~.
\]
Since $1 \le s_1 < s_2 < \dots < s_m \le n$, this can happen only 
$\log\log{n}$ times. Therefore, the amortized cost is $O(\log\log{n})$ by Lemma~\ref{lem:amortized-cost}.

\section{$O(\log^{*}n)$ Analysis}
\label{sec:logstar}

Define the potential function as
\[
\Phi\left( p \right)
:=
\frac{\size\left(p\right) }{\left(3+\log \size\left(p\right) \right)^2}
~.
\]

We first verify that the potential is increasing in $\size(p)$. This is where we need the constant $3$ (instead of $1$) in the definition.
The derivative of $\frac{x}{(3+\log x)^2}$ at $x \ge 1$ is
\begin{equation}
    \label{eqn:logstar-derivative}
    \frac{1}{(3+\log x)^2} - \frac{2}{\ln{2} (3+\log x)^3} 
    \ge \frac{1}{(3+\log x)^2} - \frac{2}{3 \ln 2(3+\log x)^2} 
    \ge \frac{1}{30(3 + \log x)^2}
    ~,
\end{equation}
where the second last inequality follows
from $\log x \ge 0$ for all $x \geq 1$. 

By the monotonicity of $\Phi(p)$ and that $\size(p)$ is nonincreasing for non-root nodes $p$ (Lemma~\ref{lem:monotonicity}), the first property of potential function analysis holds. 
Further, the second property follows by the monotonicity of $\Phi(p)$ and $\size(p) \le \sizemax(p)$.

We next verify the third property.
Consider a find step along $p_1 \to p_2 \to \cdots \to p_m \to r$, where $r$ is a root.
Suppose that the current sizes of $p_1, p_2, \dots, p_m$ are $s_1, s_2, \dots, s_m$. After the path compression, the size of $p_1$ stays the same, while the new sizes of $p_2, \dots, p_m$ are
\[
    s_{2} - s_1, s_{3} - s_2, \ldots, s_m - s_{m-1}
    ~.
\]




We next analyze the decrease of node $p_i$'s potential for $2 \le i \le m$.
%
By Equation~\eqref{eqn:logstar-derivative} and that $\frac{1}{30(3+\log x)^2}$ is decreasing in $x$, the derivative of $\frac{x}{(3+\log x)^2}$ is at least $\frac{1}{30(3+\log s_i)^2}$ for $s_i - s_{i-1} \le x \le s_i$.
Hence, when $p_i$'s size changes from $s_i$ to $s_i - s_{i-1}$, its potential decreases by at least
\[
\frac{s_{i-1}}{30 \left(3+\log s_i \right)^2}
~.
\]

Either this is at least $\frac{1}{270}$, covering the cost of edge $p_{i-1} \to p_i$, 
or we have:
\begin{equation}
    s_{i - 1} < \frac{1}{9} \left( 3+\log s_i \right)^2 = \Big( 1 + \frac{1}{3} \log s_i \Big)^2
    ~.
\label{eq:decrease}
\end{equation}


The function $\left( 1 + \frac{1}{3} \log x \right)^2$ is of higher order than $\log{x}$, but iterating it twice is less than $\log{x}$ for any $x \ge 8$.
This is easy to verify numerically,
so we defer its proof to Lemma~\ref{lem:badlogtwice}
in Appendix~\ref{app:logstar}.
Since $1 \le s_1 < s_2 < \dots < s_m \le n$, a size decrease satisfying
Equation~\eqref{eq:decrease} can happen at most
most $O(\log^{*}n)$ times along the find-path.
Hence, the amortized cost is $O(\log^{*}n)$ by Lemma~\ref{lem:amortized-cost}.

\section{Further Extensions}
\label{sec:extensions}

This section combines the ideas in Sections~\ref{sec:loglog}
and~\ref{sec:logstar} to prove better upper bounds such as $O(\log^{**}{n})$ for the amortized cost of union-find.
The following analysis uses the max sizes of the nodes. 

We first give an alternative $O(\log^{*}n)$ proof.
Consider a potential function
\[
\Phi_1 \left( p \right)
\defeq  
\frac{\sqrt{\sizemax\left(p\right)}}
{
1 + \log\sizemax\left(\parent(p)\right) 
}
~.
\]
We opt for this choice of ``nice functions'' at the cost of a large constant in the asymptotic bound. 
The constant would be much smaller if we use \[
\frac{\sizemax\left(p\right)^a}
{(1+\log\parent\left(\sizemax\left(p\right)\right))^b}
\]
with constants $a \to 1$ and $b \to 0$.


The numerator is fixed throughout.
By Lemma~\ref{lem:parent-size-monotone}, $\sizemax(\parent(p))$ is increasing over time for any non-root node $p$.
Hence the first property of potential function analysis holds. 
The second property is straightforward because of the exponent $\frac{1}{2}$. 

We next verify the third property.
Consider a find step along $p_1 \to p_2 \to \cdots \to p_m$, where $p_m=r$ is the root.
Let the max sizes of $p_1, p_2, \dots, p_m$ at the moment be $s_1, s_2, \dots, s_m$.
Consider any $p_i$ where $1 \le i \le m-2$.
The potential of this node changes from
\[
\frac{\sqrt{s_i}}{1 + \log s_{i + 1}}
\]
to
\[
\frac{\sqrt{s_i}}{1 + \log s_{m}}
\]
because 
its parent becomes $p_m$ after the path compression.


Since $\log s_{i+1} \le \log s_{i+2} - 1 \le \log s_m - 1$ by Lemma~\ref{lem:geometric-decrease}, the decrease is
\[
\sqrt{s_{i}}
\left( \frac{1}{1 + \log s_{i+1}}
- \frac{1}{1 + \log s_m} \right)
\geq
\sqrt{s_{i}}
\left( \frac{1}{1 + \log s_{i+1}}
- \frac{1}{2+\log s_{i+1}} \right)
\ge
\frac{\sqrt{s_{i}}}{2 \left( 1 + \log s_{i+1} \right)^2}
~.
\]
Either this is at least $\frac{1}{2^{13}}$, covering the cost of edge $p_i \to p_{i+1}$, or we have
\begin{equation}
\label{eqn:layer-1}
s_{i}
\leq
\left( \frac{1 + \log s_{i + 1}}{8} \right)^4
~.
\end{equation}

Recall that the $\log$ function represents the number of times a node's value
can halve, and the $\log^*$ function represents the number of times we can take $\log$.
We define a variant of it for our modified iteration of taking $\log$ and raising to the fourth power:
\[
\logs\left( x \right)
\defeq
\min_{k}
\left[
\left( y \leftarrow \left( \frac{1 + \log y}{8} \right)^4 \right)^{\left( k \right)}
\left( x \right)
\leq
1
\right]
~.
\]
That is, the number of times to iteratively apply
$x \leftarrow (\frac{1 + \log x}{8})^4$
until the result is at most $1$.
This is asymptotically the same as $\log^{*}$.
Intuitively, this is because raising to the fourth power is of lower order compared to the inverse of logarithm (i.e., the exponential function).
Hence, the value still decreases almost as fast as taking logarithm.
We defer the proof to Lemma~\ref{lem:logs} of Appendix~\ref{sec:logstar}.

By the definition of $\logs$ and that $1 \le s_1 < s_2 < \dots < s_m \le n$, Equation~\eqref{eqn:layer-1} holds for at most $O\big(\logs n\big)$ different $p_i$'s.
Hence, the amortized cost is $O(\logs n) = O(\log^* n)$ by Lemma~\ref{lem:amortized-cost}.

We next define the level-$2$ potential of a node:
\[
\Phi_2 \left( p \right)
\defeq
\sqrt{\sizemax\left( p \right)}
\cdot
\left(
\frac{1}{1 + \log \left( \sizemax \left( \parent \left( p \right)\right) \right)}
+
\frac{1}{1 + \logs \left( \sizemax \left( \parent \left( p \right)\right) \right)}
\right).
\]
Note that it is critical we use $\sqrt{\sizemax(p)}$ instead
of $\sizemax(p)$ in the numerator:
the $\logs$ of a node's ancestors' sizes do not decrease fast enough to allow telescoping as in the proof of Lemma~\ref{lem:sizemax-upper-bound} (even if we raise the numerator to some power $d > 1$ like in Section~\ref{sec:logstar}).

The first two properties of the potential function analysis still hold for the same reason as for the first potential in the section. For the third property, we consider a path compression along $p_1 \to p_2 \to \cdots \to p_m$
with max sizes $s_1, s_2, \ldots, s_m$ respectively.
The above analysis shows that the decrease of the first term of the potential pays for all edges along the path except for (1) $p_{m-1} \to p_m$ and (2) edges $p_i \to p_{i+1}$ satisfying $s_i \le (\frac{1 + \log s_{i+1}}{8})^4$.

It remains to use the decrease of the second term of the potential to pay for the second subset of edges.
Since such an edge $p_i \to p_{i+1}$ satisfies $s_i \le (\frac{1 + \log s_{i+1}}{8})^4$, we have
%
\[
\logs \left(s_{i} \right)
\leq
\logs \left(s_{i+1} \right) - 1
\leq
\logs \left(s_{m} \right) - 1
~.
\]Hence, the decrease in potential function in the second term (of node $p_{i-1}$) is
at least
\begin{align*}
\sqrt{s_{i - 1}}
\left( \frac{1}{1 + \logs s_{i}}
- \frac{1}{1 + \logs s_{m}} \right)
\geq
\sqrt{s_{i - 1}} \cdot
\frac{1}{2\left( 1 + \logs s_{i} \right)^2}
~.
\end{align*}
Either this is at least $\frac{1}{2^{13}}$, covering the cost of edge $p_i \to p_{i+1}$, or we have
\begin{equation}
\label{eqn:layer-2}
s_{i - 1} \leq \left( \frac{1 + \logs s_{i}}{8} \right)^4
~.
\end{equation}
We now define
\[
\logss\left( x \right)
\defeq
\min_{k}
\left[
\left( y \leftarrow \left( \frac{1 + \logs y}{8} \right)^4 \right)^{\left( k \right)}
\left( x \right)
\leq
1
\right].
\]
By the definition of $\logss$ and that $1 \le s_1 < s_2 < \dots < s_m \le n$, Equation~\eqref{eqn:layer-2} holds for at most $O\big(\logss n\big)$ different $p_i$'s.
Hence, the amortized cost is $O(\logss n)$ by Lemma~\ref{lem:amortized-cost}.

We can further define the third-level potential function by adding a $\frac{1}{1+\logss}$ term, and so forth.


\section{Proof Using the Original Definition of Ackermann Functions}
\label{sec:direct}

This section directly considers the original definition of the Ackermann functions. 
Recall that the Ackermann functions $A_k: \mathbb{N} \to \mathbb{N}$, $k \ge 0$, are recursively defined as follows
\[
    A_0(\ell) = \ell+1 \quad,\qquad
    A_{k+1}(\ell) = A_k^{(\ell+1)}(1) ~.
\]

We will consider $r(p) = \left\lfloor \log \sizemax(p) \right\rfloor$
%
with two properties:
(1) $r(\parent(p)) \ge r(p) + 1$ for any non-root node $p$, and 
(2) $r(\parent(p)) = O \big( \frac{\sizemax(p)}{(1+\log\sizemax(p))^2} \big)$.
We remark that it suffices to use $p$'s rank at the end for $r(p)$ in union-by-rank.
%
%
%
For any non-root node $p$, define its potential as
\[
\Phi(p)
~\defeq~
\sum_{k = 0}^{\alpha(n)}
~ \sum_{\ell = 1}^{ r\left(p\right)}
~ \one \left(A_k^{(\ell+1)}\left( r\left( p \right) \right)
>
r\left(\parent\left( p \right) \right)
\right)
~.
\]

It is similar to the potential function in the last section in the sense that it is also increasing in $\sizemax(p)$ (through $r(p)$), which is fixed throughout, and is also decreasing over time for any non-root node $p$ because $\sizemax(\parent(p))$ and thus $r(\parent(p))$ is nondecreasing (Lemma~\ref{lem:parent-size-monotone}).

As a result, the first property (monotonicity) of potential function analysis holds.
%
The second property (boundedness) also holds because by definition $\Phi(p) \le (\alpha(n) + 1) r(p)$ and thus
\[
    \frac{\Phi(p)}{\alpha(n)}\le O(r(p)) = O(\log \sizemax(p))
    ~.
\]

For the third property, consider a path compression along $p_1 \to p_2 \to \cdots \to p_m$, where $p_m=r$ is the root.
For any node $p_i$, $1 \le i < m$, with parent $p_{i+1}$,
let $k_i$ be the largest $k \ge 0$ for which %
\[
    A_k(r(p_i)) \le r(p_{i+1})
    ~.
\]
This is well-defined since
$A_0\left( r\left( p_i \right) \right)
=
r\left( p_i \right) + 1
\le
r\left( p_{i+1} \right)$. 
Further, we have $k_i \le \alpha(n)$
because $A_{\alpha(n)+1}(1) > n \ge \sizemax(p_{i+1}) \ge r(p_{i+1})$ by the definition of inverse Ackermann function $\alpha(n)$.

Next, we let $\ell_i$ be the largest $\ell \ge 1$ for which 
\[
A_{k_i}^{(\ell)}\left(r\left( p_{i} \right) \right)
\le
r\left( p_{i + 1} \right)
~.
\]
%
We have $\ell_i \le r(p_i)$, because by our choice of $k_i$ 
\[
    r(p_{i+1}) < A_{k_i+1} \big( r(p_i) \big) = A_{k_i}^{(r(p_i)+1)}(1) \le A_{k_i}^{(r(p_i)+1)}\big( r(p_i) \big) ~.
\]

For any $0 \le k \le \alpha(n)$, consider all $p_i$'s with $k_i = k$.
We next show that all these nodes, except the one closest to the root, have their potential function values decreased by at least $1$.
Consider any such node $p_i$.
Its old parent $p_{i+1}$ satisfies
\[
    A_k^{\ell_i}\big(r(p_i)\big) \le r(p_{i+1}) < A_k^{\ell_i+1}\big(r(p_i)\big)
    ~.
\]
Further, since it is not the closest to the root among these vertices, there is $j > i$ such that $k_j = k$, which means $r(p_{j+1}) \ge A_k(r(p_j))$.
Noting that $r(p_m) \ge r(p_{j+1})$ and $r(p_j) \ge r(p_{i+1})$, we have
\[
    r(p_m) \ge A_k\big(r(p_{i+1})\big) \ge A_k^{\ell_i+1}\big(r(p_i)\big)
    ~.
\]
Therefore, changing $p_i$'s parent from $p_{i+1}$ to $p_m$ makes the indicator for $k = k_i$ and $\ell = \ell_i$ change from $1$ to $0$ in $p_i$'s potential function.

In sum, for each $0 \le k \le \alpha(n)$, there is at most one node with $k_i = k$ whose potential function fails to decrease.
The amortized cost is therefore $O(\alpha(n))$ by Lemma~\ref{lem:amortized-cost}.

\section{Conclusion}
\label{sec:conclusion}

We gave potential-function-based analyses of the union-find
data structure that are closer to other amortized potential
functions used to analyze tree-based data structures.
These proofs differ from existing proofs in that they no longer
require explicit definitions of the Ackermann function.
Nonetheless, they are able to naturally interpolate to the
optimal $O(\alpha(n))$ bound using the $k$th-iterated logs based characterization of
the inverse Ackermann function.

From a presentation perspective, the authors believe that only
the $O(\log\log{n})$ and $O(\log^{*}n)$ proofs 
(Sections~\ref{sec:loglog} and~\ref{sec:logstar} respectively)
are significantly simpler than proofs via the Ackermann
function.
Our current attempts at extending these proofs beyond $\log^{*}n$
require potential functions involving both a node and its parent,
and are essentially continuous variants of the proofs that
explicitly define the inverse Ackermann function.
We pose as an open question whether potential functions
based solely on $\size(p)$ can also give bounds of
$\log^{**}n$ or better.

\bibliographystyle{alpha}
\bibliography{ref}

\begin{appendix}

\section{Ackermann Function and Its Alternate Characterizations}
\label{sec:ackermann}

First, recall that the Ackermann functions $A_k: \mathbb{N} \to \mathbb{N}$, $k \ge 0$, are recursively defined as follows
\[
    A_0(\ell) = \ell+1 \quad,\qquad
    A_{k+1}(\ell) = A_k^{(\ell+1)}(1) ~.
\]

Further recall that superscript $*$ means the number of iterative applications of a function until the value is at most $1$.
We will write $f^{* \times k}$ for function
\[
    f^{\overbrace{** \dots *}^{\text{$k$ times}}}
    ~.
\]

\begin{lemma}
The inverse Ackermann function $\alpha(n)$ is within a constant factor of the minimum natural number $k$ such that $\log_2^{* \times k}(n) = 1$.
\begin{proof}
Define $B_k(x) = \min\{i \in \mathbb{N} : A_k(i) \ge x\}$.
By definition, we have
\[
    B_k(A_k(x)) = x ~,\quad
    A_k(B_k(x)) \ge x ~.
\]
Hence,
$$B_{k+1}(x) = \min\{i \in \mathbb{N} : A_{k+1}(i) \ge x\} = \min\{i \in \mathbb{N} : A_k^{(i+1)}(1) \ge x\} = \min\{i \in \mathbb{N} : 1 \ge B_k^{(i+1)}(x)\}$$
where the last conditions imply each other by applying $A_k^{i+1}$ or $B_k^{i+1}$ to both sides. 
In particular, we get that $B_{k+1}$ equals $B_k^*(x)-1$ (unless $B_k^*(x)$ was zero, in which case $B_{k+1}$ would be zero).

The first few Ackermann functions are
\[
A_0(x) = x+1 ~,\quad A_1(x) = x+2 ~,\quad A_2(x) = 2x+3 ~,\quad A_3(x) = 2^{x+3} - 3
\]
each of which can be easily verified by induction. 
Therefore
\[
    B_3(x) = \min\{i \in \mathbb{N} : 2^{i+3}-3 \ge x\} = \max \big\{ \lfloor \log_2(x+3) \rfloor - 3, 0 \big\}
    \le \log_2 x
    ~.
\]
Combining this and $B_{k+1}(x) = \max \{ B_k^*(x) - 1, 0 \} \le B_k^*(x)$, we have
\[
    B_{k+3}(x) \le \log_2^{* \times k}(x)
\]
whenever the right-hand-side is greater than $1$.

To attain the other direction, first note that if
$f : \mathbb N^+ \to \mathbb N$
satisfies
$f(x) < x$ for all $x$
then $f(x) > f^*(x)$ when $f(x) > 1$ and otherwise $f^*(x) \le 1$.
Applying this we get
\[
B_3\left(x\right)
=
\left\lfloor \log_2\left(x + 3\right)\right\rfloor
-
3 \ge \left \lfloor \log_2(x) \right \rfloor
-
3 \ge \log_2^{***}\left(x\right)
\]
when the right hand side is greater than one
--- and the $^*$ operator does consider what values that function takes on $0$ or $1$.
Going further we can induct on $k$ to get
\[
B_{k+3}\left(x\right)
\ge
B_{k+2}^{*}\left(x\right)
-
1
\ge
B_{k+2}^{**}\left(x\right)
\ge
\left(\log_2^{* \times \left(2k+1\right)}\right)^{**}\left(x\right)
=
\log_2^{* \times \left(2k+3\right)}\left(x\right)
\]
whenever the right hand side is greater than one.

By definition,
$$\alpha(x) = \min\{i \in \omega : A_i(1) \ge x\} = \min\{i \in \omega : B_i(x) \le 1\}$$
and let $\alpha'(x) = \min\{i \in \omega : \log_2^{* \times k}(x) = 0 \}$. From the bounds above if $\log_2^{* \times (2k+3)}(x) > 1$ then $B_{k+3}(x) > 1$ so $\alpha' = O(\alpha)$ and if $B_{k+3}(x) > 1$ then $\log_2^{* \times k}(x) > 1$ so $\alpha = O(\alpha')$.
\end{proof}
\end{lemma}

\section{Comparing Iterated
Modified Functions with $\log^{*}$}
\label{app:logstar}

We check that our lower-order modifications
to the inputs of logarithm still give iteration
counts comparable to iterated logarithms.

We first prove that iterating function $( 1 + \frac{1}{3} \log x )^2$ twice is smaller than $\log{x}$ for any $x \ge 8$.

\begin{lemma}
\label{lem:badlogtwice}
Consider the function $f(x) = ( 1 + \frac{1}{3} \log x )^2$.
For all $x \geq 8$, we have
\[
f\left( f \left( x \right) \right) \leq \log{x}.
\]
\end{lemma}
%

\begin{proof}
Let $y = \log x$.
Then $y \geq 3$, and $y$ is monotonically increasing in $x$.
The condition that we want to prove is equivalent to
\[
    y - \Big( 1 + \frac{2}{3} \log \big( 1 + \frac{y}{3} \big) \Big)^2 \ge 0
    ~.
\]

This is true for $y = 3$, because $3 - (1+\frac{2}{3})^2 = \frac{2}{9} > 0$.
It remains to verify that the left-hand-side is increasing.
The left-hand-side's derivative is:
\begin{align*}
    1 - 2 \Big( 1 + \frac{2}{3} \log \big( 1 + \frac{y}{3} \big) \Big) \frac{2}{3 \ln 2(3+y)}
    &
    = 
    \frac{3\ln 2(3+y) - 4 \left( 1 + \frac{2}{3} \log ( 1 + \frac{y}{3} ) \right) }{3\ln 2(3+y)} \\
    &
    = \frac{(9 \ln 2 - 4) + (3 \ln 2) y - \frac{8}{3} \log ( 1 + \frac{y}{3} )}{3\ln 2(3+y)} \\
    &
    \ge \frac{(9 \ln 2 - 4) + (3 \ln 2 - \frac{8}{9 \ln 2}) y}{3\ln 2(3+y)} > 0
    ~,
\end{align*}
where the second last inequality follows by $\ln(1+z) \le z$.
\end{proof}

We now turn to the iterated variants of such function
used in Section~\ref{sec:extensions}.

\begin{lemma}
\label{lem:logs}
The function 
\[
\logs\left( x \right)
=
\min_{k}
\left[
\left( y \leftarrow \left( \frac{1 + \log y}{8} \right)^4 \right)^{\left( k \right)}
\left( x \right)
\leq
1
\right]
\]
satisfies $\logs(x) \le 2 \log^{*}(x)$.
\end{lemma}

\begin{proof}
%
%
%
%
We first verify that $\logs(x)$ is well-defined, by proving for any $y \ge 2$ that
\[
    \left( \frac{1 + \log y}{8} \right)^4 \le \frac{y}{2}
    ~.
\]
Taking logarithms on both sides and changing variables with $z = \log y$, this is equivalent to proving for any $z \ge 1$ that
\[
    4 \big( \log (1+z) - 3 \big) \le z - 1
    ~.
\]
Rearranging terms, we can further rewrite it as
\[
    \log \left( \frac{1+z}{2^\frac{11}{4}} \right) \le \frac{z}{4}
    ~.
\]
This follows because the left-hand-side is at most
\[
    \log \left( 1 + \frac{z}{2^\frac{11}{4}} \right) \le \frac{z}{2^{\frac{11}{4}} \ln 2} < \frac{z}{4}
    .
\]

It remains to very that applying function $y \leftarrow (\frac{1 + \log{y}}{8})^4$ twice yields a value that is at most $\log y$ for any $y \ge 2$.
Changing variables with $z = \log y$, we need to prove for any $z \ge 1$ that
\[
    \left( \frac{1 + 4 \log \frac{1+z}{8}}{8} \right)^4 \le z
    ~.
\]
Relaxing $1 + z$ to $2z$ and taking fourth root on both sides, it suffices to prove
\[
    \frac{4 \log z - 7}{8} \le \sqrt[4]{z}
    ~.
\]
Changing variables with $u = \sqrt[4]{z}$, this is equivalent to
\[
    2 \log u - \frac{7}{8} \le u
    ~.
\]
This holds because $2 \log u - u$ achieves maximum value approximately $0.172 < \frac{7}{8}$ at $u = \frac{2}{\ln 2}$.
%
%
\end{proof}

\end{appendix}

\end{document}